\documentclass[10 pt]{IEEEtran}
\usepackage{makecell}
\usepackage{tabularx}
\usepackage{graphics}
\usepackage{graphicx}
\usepackage{amsmath,amsfonts,amssymb,varioref,mathrsfs,verbatim}
\usepackage{latexsym}
\usepackage{epsfig}
\usepackage{subcaption}
\usepackage{array}
\usepackage[english]{babel}
\usepackage{color}
\usepackage{babel}

\newtheorem{lemma}{Lemma}

\renewcommand{\figurename}{Fig.}
\renewcommand{\figurename}{Fig.}
\addto\captionsenglish{\renewcommand{\figurename}{Fig.}}

\makeatletter
\let\l@ENGLISH\l@english
\makeatother
\begin{document}
\title{On the Security of Millimeter Wave  \\ Vehicular  Communication Systems using\\  Random Antenna Subsets}
\author{Mohammed E. Eltayeb, Junil Choi, Tareq Y. Al-Naffouri$^*$, and Robert W. Heath, Jr. \\
The University of Texas at Austin, Emails: \{meltayeb,junil.choi,rheath\}@utexas.edu.\\
$^*$King Abdullah University of Science and Technology, Email: tareq.alnaffouri@kaust.edu.sa
}
\maketitle

\begin{abstract}
Millimeter wave (mmWave) vehicular communication systems have the potential to improve traffic efficiency and safety. Lack of secure communication links, however, may lead to a formidable set of abuses and attacks. To secure communication links, a physical layer precoding technique for mmWave vehicular communication systems is proposed in this paper. The proposed technique exploits the large dimensional antenna arrays available at mmWave systems to produce direction dependent transmission. This results in coherent transmission to the legitimate receiver and artificial noise that jams eavesdroppers with sensitive receivers. Theoretical and numerical results demonstrate the validity and effectiveness of the proposed technique and show that the proposed technique provides high secrecy throughput when compared to conventional array and switched array transmission techniques.

\end{abstract}

\section{Introduction} \label{sec:Intro}


Future intelligent transportation systems (ITSs) will rely on vehicular communication technologies to enable a variety of applications for safety, traffic efficiency, driver assistance, and infotainment \cite{r1}, \cite{r2}.  The abundance of bandwidth in the  millimeter wave (mmWave) will enable higher data rate communication between vehicles to exchange the raw data from LIDAR, radar, and other sensors to support advanced driver-assisted and safety-related functionalities.

Like regular communication systems, vehicular communication systems are prone to a set of abuses and attacks that could jeopardize the efficiency of transportation systems and the physical safety of vehicles and drivers. Security threats and attacks in vehicular environments have been described in detail in \cite{r2}-\cite{r7}. To preserve the privacy and security of vehicular communication networks, higher-layer encryption techniques have been proposed in \cite{r7}-\cite{r9}, and references therein. These techniques, however, are based on digital signature methods that require vehicles to store large number of encryption keys. These keys must also be regularly exchanged, hence creating, in addition to the processing overhead, an additional communication overhead \cite{r1}, \cite{r2}. With the emergence of new time-sensitive ITS applications and the increasing size of the decentralized vehicular network, the implementation of these higher-layer techniques becomes complex and challenging. Additionally these techniques fail to secure communication links if the encryption keys are compromised, and/or key distribution becomes difficult.

To mitigate these challenges, \emph{keyless} physical layer (PHY) security techniques \cite{dm401}-\cite{dm5}, which do not directly rely on upper-layer data encryption or secret keys, can be employed to secure vehicular communication links. These techniques are based on the use of one or more of the following approaches: (i) directional beamforming \cite{dm402}-\cite{db2}, (ii) distributed arrays (or nodes) \cite{anr1} \cite{dm403}, (iii) dual-beam techniques \cite{dual}, and (iv) switched arrays \cite{s1}-\cite{dm5}. In all of these techniques, the transmitter uses multiple antennas to create interference or \textit{artificial noise} within the null space of the receiver. Such interference will degrade the channel of potential eavesdroppers and will not impact the receiver's channel.

Despite their effectiveness, the techniques in \cite{dm402}-\cite{dual}  require fully digital baseband precoding, which is difficult to realize in mmWave vehicular systems. This is due to the limited number of available RF chains and additional hardware and power constraints in mmWave systems \cite{pi2011} \cite{ahmed}. While it might be justifiable to deploy large number of RF chains in cellular base stations with relaxed power constraints, this leverage might not be possible in vehicular transmitters. Due to both power and hardware constraints, switched array techniques \cite{s1}-\cite{dm5} might be suitable for vehicular mmWave systems with large number of antennas. These techniques use antenna switches to associate a subset of antennas with every transmission symbol. This results in coherent transmission towards the receiver and induces artificial noise in all other directions. Nonetheless, it was recently shown that switched array security techniques can be breached  by exploiting the sparsity in the transmit antenna array \cite{attack1} \cite{attack2}.  


In this paper, we propose a new PHY security  technique for vehicular communications that does not require the exchange of keys between vehicles. We assume a single transmitter communicating with a single receiver in the presence of an eavesdropper as shown in Fig. \ref{fig:mod4}.  The proposed technique uses an antenna architecture with a single RF chain (instead of multiple RF chains), and performs analog precoding with antenna selection. To transmit an information symbol, a random set (called a subset) of antennas are co-phased to transmit the information symbol to the receiver, while the remaining antennas
are used to randomize the far field radiation pattern at non-receiver directions. This results in coherent transmission to the receiver and a noise-like signal that jams potential eavesdroppers with sensitive receivers. Note that the proposed technique is robust to plaintext attacks and other attacks highlighted in \cite{attack1} and \cite{attack2} since a random subset is associated with every data symbol, and all antennas are used for data transmission, i.e. the array is no longer sparse.


	\begin{figure}
		\begin{center}
\includegraphics[width=3in]{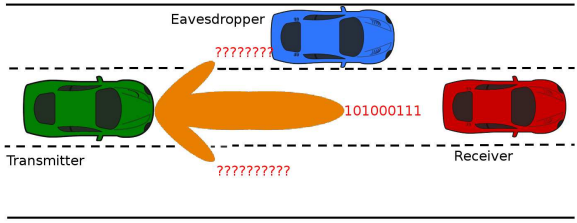}
				\caption{Vehicle to vehicle communication with a possible eavesdropper. The transmitting vehicle  is communicating with a target vehicle  while the eavesdropper tries to intercept the transmitted data.}
			\label{fig:mod4}
		\end{center}
	\end{figure}

\section{System Model} \label{sec:Model}

We consider a mmWave system with a transmitting vehicle (transmitter) and a receiving vehicle (receiver) in the presence of a single or multiple eavesdroppers as shown in Fig. \ref{fig:mod4}. The transmitter is equipped with $N_{\text{T}}$ transmit antennas and a single RF chain as shown in Fig. \ref{fig:ant}. The receiver is equipped with $N_{\text{R}}$ antennas. We adopt a uniform linear array (ULA) with isotopic antennas along the x-axis with the array centered at the origin; nonetheless, the proposed technique can be adapted to arbitrary antenna structures. Since the array is located along the x-y plane, the receiver's location is specified by the azimuth angle of arrival/departure (AoA/AoD). The transmitter is assumed to know the angular location of the target receiver but not of the potential eavesdroppers.

The transmit data symbol $s(k)\in \mathbb{C}$, where $\mathbb{E}[|s(k)|^2]=1$ and $k$ is the symbol index,  is multiplied by a unit norm transmit beamforming vector $\mathbf{f} = [f_1 \quad f_2 \quad ... \quad f_{N_\text{T}}]\in \mathbb{C}^{N_{\text{T}}} $ with $f_n$ denoting the complex weight on transmit antenna $n$. At the receiver, the received
signals on all antennas are combined with a unit norm receive combining vector $\mathbf{w} = [w_1 \quad w_2 \quad ... \quad w_{N_\text{R}}]\in \mathbb{C}^{N_{\text{R}}} $.  Assuming a narrow band line-of-sight (LOS) channel with perfect synchronization, the received signal can be written as
\begin{eqnarray}\label{c1mimo}
y(k,\phi,\theta) = \sqrt{P \alpha}\mathbf{w}^*\mathbf{H}(\phi,\theta)\mathbf{f}s(k) + z(k),
\end{eqnarray}
where $\phi$ is AoA at the receiver and $\theta$ is AoD from the transmitter, $P$ is the transmission power and $\alpha$ is the path-loss. In (\ref{c1mimo}), $\mathbf{H}$ is the ${N_\text{R}\times N_\text{T}}$ channel matrix that represents the mmWave channel between the transmitter and the receiver, and $z(k) \sim  \mathcal {CN}(0,\sigma^2)$ is the additive noise. 

Due to the dominant reflected path from the road surface, a two-ray model is usually adopted in the literature to model LOS vehicle-to-vehicle
communication \cite{cv1}-\cite{cv4}. Based on this model, the channel $\mathbf{H}$ can be modeled using array manifold concepts as $\mathbf{H}~=~g\mathbf{a}_\mathrm{r}(\phi)\mathbf{a}^*_\mathrm{t}(\theta)$,  where $g = \frac{1}{\sqrt{2}}(1+re^{j\Phi})$, $r=-1$ is the reflection coefficient of the road,
and $\Phi$ is the angle between the direct and the reflected path and it is given by $\Phi = \frac{2\pi}{\lambda} \frac{2h_\text{t}h_\text{r}}{D}$ \cite{cv4}. The variables $h_\text{t}$, $h_\text{r}$ represent the height of the transmitting and receiving antennas respectively, $D$ is the distance between the vehicles, and  $\lambda$ is the wavelength.  {{Since $h_\text{t} h_\text{r} \ll D$, which is typical in vehicular environments, the angle $\Phi$ is small and both the LOS and non-LOS paths combine at the azimuth AoA $\phi$.}}  The vector $ \mathbf{a}_\mathrm{r}(\phi)$ represents the receiver's array manifold corresponding to the AoA $\phi$ and the vector $\mathbf{a}_\mathrm{t}(\theta)$  represents the transmitter's array manifold corresponding to the AoD $\theta$. For simplicity, the receiver's beam is assumed to be aligned to towards the transmitter for maximum reception {{i.e. $\mathbf{w} = \mathbf{a}_\mathrm{r}(\phi)$. Hence, the received signal along an angle $\theta$ becomes
\begin{eqnarray}
y(k,\theta) &=&  \sqrt{P \alpha} \mathbf{a}^*_\mathrm{r}(\phi) \mathbf{H}(\phi,\theta) \mathbf{f}s(k) + z(k) \\ \label{c1}
&=&  \sqrt{P \alpha N_\text{R}}\mathbf{h}^*(\theta)\mathbf{f}s(k) + z(k),
\end{eqnarray}  
where }} $N_\text{R}$ is the array gain at the receiver. For a ULA with $d \leq \frac{\lambda}{2}$ antenna spacing, the channel for a receiver located along
the $\theta$ direction can be written as $\nonumber \mathbf{h}^*(\theta) = g [e^{-j( \frac{N-1}{2}) \frac{2\pi d}{\lambda} \cos \theta} \quad e^{-j( \frac{N-1}{2}-1) \frac{2\pi d}{\lambda} \cos \theta} \quad \\ ... \quad e^{j( \frac{N-1}{2}) \frac{2\pi d}{\lambda} \cos \theta}   ]$ \cite{trees}.

	\begin{figure}
		\begin{center}
\includegraphics[width=2.2in]{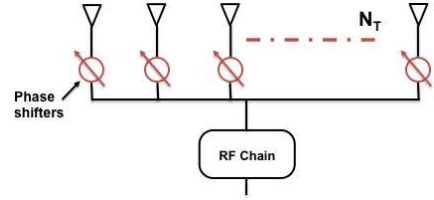}
				\caption{Array architecture with $N_\text{T}$ antennas and a single RF chain.}
			\label{fig:ant}
		\end{center}
	\end{figure}

\section{Millimeter wave Secure Transmission } \label{sec:ST}
In this section, we introduce a novel analog PHY security technique that randomizes the data stream at the eavesdropper without the need for a fully digital array or antenna switches as done in \cite{dm4}-\cite{dm5}.  Instead of using all antennas for beamforming,  a set of random antennas are used for coherent beamforming, while the remaining antennas  are set to destructively combine at the receiver. The indices of these antennas are randomized in every
symbol transmission. This randomizes the beam pattern sidelobes and as a result, produces noise-like signals at non-receiver directions. Although the target receiver would observe a fixed gain reduction (due to destructive combining), malicious eavesdroppers will observe a noise-like interference. 

Let $\mathcal{I}_M(k)$ be a random subset of $M$ antennas used to transmit the $k$th symbol, $\mathcal{I}_L(k)$ be a subset that contains the indices of the remaining antennas, $\mathcal{E}_L(k)$ be a subset that contains the even entries of   $\mathcal{I}_L(k)$, and $\mathcal{O}_L(k)$ be a subset that contains the odd entries of $\mathcal{I}_L(k)$.  The progressive inter-antenna phase shifts are set as
  \begin{equation}\label{efbp}
\Upsilon_{n}(k)  = \left\{
               \begin{array}{ll}
               \Psi_{n}(k), & \hbox{  $n\in \mathcal{I}_M(k)$  } \\
                   \Psi_{n}(k), & \hbox{  $n\in \mathcal{E}_L(k)$  }  \\
                  -j\log(-e^{ j \Psi_{n}(k)}), & \hbox{  $n\in \mathcal{O}_L(k)$  } \\
               \end{array}
               \right.
\end{equation}
where $\Psi_{n}(k) = { \left(\frac{N-1}{2}-n \right)2\pi \frac{d}{\lambda} \cos(\theta_{\text{R}})}$, and $\theta_{\text{R}}$ is the transmit direction towards the legitimate receiver. Note in (\ref{efbp}) $e^{j( -j\log(-e^{ j \Psi_{n}(k)}))} = -e^{j\Psi_{n}(k)}$, and the $n$th entry of the beamforming vector $\mathbf{f}(k)$ is $f_n(k)=\frac{1}{\sqrt{N_\text{T}}}e^{ j \Upsilon_{n}(k)}$. Using (\ref{c1}) and (\ref{efbp}), the received signal at the receiver along $\theta_\text{R}$ becomes
\begin{eqnarray}\label{yra8}
\begin{aligned} 
 & \hspace{-1mm} y_{\text{R}}(k,\theta_{\text{R}})  =   \sqrt{\frac{P\alpha N_\text{R}}{{N_\text{T}}}} g_\text{R} s(k) \times  \\
  &  \bigg( \sum_{m \in \mathcal{I}_M(k)} e^{-j\left(\frac{N-1}{2}-m\right)  \frac{2\pi d}{\lambda}  \cos\theta_\text{R} }  e^{j \left(\frac{N-1}{2}-m\right)  \frac{2\pi d}{\lambda}  \cos\theta_\text{R} }  +  \\  & \hspace{-0.5mm} \sum_{n \in \mathcal{E}_L(k)} e^{-j \left(\frac{N-1}{2}-n\right)  \frac{2\pi d}{\lambda}  \cos\theta_\text{R} }  e^{j \left(\frac{N-1}{2}-n\right)  \frac{2\pi d}{\lambda}  \cos\theta_\text{R} } - \\ \hspace{-3mm} & \sum_{n \in \mathcal{O}_L(k)} \hspace{-3mm} e^{-j \left(\frac{N-1}{2}-n\right)  \frac{2\pi d}{\lambda}  \cos\theta_\text{R} } e^{j \left(\frac{N-1}{2}-n\right)  \frac{2\pi d}{\lambda}  \cos\theta_\text{R} }\bigg)  + z (k),
\end{aligned}
\end{eqnarray}        
where $g_\text{R}$ is the channel gain at the receiver, $N_\text{R}$ is the receiver's array gain, and $N_\text{T}$ is a power normalization term. Grouping similar terms in (\ref{yra8}) we obtain         
 \begin{eqnarray}\label{y1b1s}
 \begin{aligned} 
 & \hspace{-5mm} y_{\text{R}}(k,\theta_{\text{R}}) =  s(k) g_\text{R} \sqrt{\frac{P \alpha  N_\text{R}}{{N_\text{T}}}} \times \\ & \bigg(  \sum_{m \in \mathcal{I}_M(k)}  e^{j \left(\frac{N-1}{2}-m\right)  \frac{2\pi d}{\lambda} (\cos\theta_\text{R}- \cos\theta_\text{R}) } \\ &+ \sum_{n \in \mathcal{E}_L(k)}  e^{j \left(\frac{N-1}{2}-n\right)  \frac{2\pi d}{\lambda} (\cos\theta_\text{R}- \cos\theta_\text{R}) }  \\  &- \sum_{n \in \mathcal{O}_L(k)}  e^{j \left(\frac{N-1}{2}-n\right)  \frac{2\pi d}{\lambda} (\cos\theta_\text{R}- \cos\theta_\text{R}) }  \bigg) + z (k)    \\
&=   \underbrace{s(k)}_{\substack{\text{ information }\\\text{symbol}}}  \underbrace{g_\text{R}\sqrt{P\alpha}}_{\substack{\text{effective channel }\\\text{gain}}} \underbrace{M\sqrt{\frac{N_\text{R}}{{N_\text{T}}}}}_{\text{array gain}}  +  \underbrace{z(k)}_{\substack{\text{additive }\\\text{noise}}}.
 \end{aligned}
 \end{eqnarray}
Similarly, the received signal at the eavesdropper can be written as
\begin{eqnarray}\label{arnl}
\begin{aligned} 
 &y_{\text{E}}(k,\theta_{\text{E}}) =   \sqrt{\frac{P\alpha_\text{E} N_\text{E}}{{N_\text{T}}}}g_\text{E} s(k) \times  \\&  \bigg( \sum_{m\in \mathcal{I}_M(k)} e^{-j\left(\frac{N-1}{2}-m\right)  \frac{2\pi d}{\lambda}  \cos\theta_\text{E} }  e^{j \left(\frac{N-1}{2}-m\right)  \frac{2\pi d}{\lambda}  \cos\theta_\text{R} }  + \\ &  \sum_{n \in \mathcal{E}_L(k)}   e^{-j \left(\frac{N-1}{2}-n\right)  \frac{2\pi d}{\lambda}  \cos\theta_\text{E} } e^{j \left(\frac{N-1}{2}-n\right)  \frac{2\pi d}{\lambda}  \cos\theta_\text{R} } - \\   & \sum_{n \in \mathcal{E}_L(k)}  \hspace{-1mm}   e^{-j \left(\frac{N-1}{2}-n\right)  \frac{2\pi d}{\lambda}  \cos\theta_\text{E} } e^{j \left(\frac{N-1}{2}-n\right)  \frac{2\pi d}{\lambda}  \cos\theta_\text{R} } \hspace{-1mm}  \bigg)  \hspace{-1mm} +\hspace{-1mm}  z_\text{E}(k)\\ 
    &=  \underbrace{s(k)}_{\substack{\text{ information }\\\text{symbol}}}  \underbrace{g_\text{E}\sqrt{P\alpha_\text{E} }}_{\substack{\text{effective channel }\\\text{gain}}} \underbrace{\sqrt{N_\text{E}}\beta}_{\substack{\text{artifical }\\\text{noise}}}  +  \underbrace{z_\text{E}(k)}_{\substack{\text{additive }\\\text{noise}}},  
   \end{aligned}
   \end{eqnarray}
where $\alpha_\text{E}$ and $g_\text{E}$ are the path loss and the channel gain at the eavesdropper, $N_\text{E}$ is the eavesdropper's array gain,  and $z_\text{E} \sim \mathcal{CN} (0,\sigma_\text{E}^2)$ is the additive noise at the eavesdropper. The term $\beta$ in (\ref{arnl}) is
\begin{eqnarray}
\nonumber  \beta &=& \sqrt{\frac{1}{{N_\text{T}}}} \bigg (\sum_{m\in \mathcal{I}_M(k)} e^{j \left(\frac{N-1}{2}-m\right)  \frac{2\pi d}{\lambda} (\cos\theta_\text{R} - \cos\theta_\text{E})} + \\ \nonumber  &&  \sum_{n \in \mathcal{E}_L(k)}  e^{j \left(\frac{N-1}{2}-n\right)  \frac{2\pi d}{\lambda}  \cos(\theta_\text{R}-\theta_\text{E}) }   - \\ \label{anb}  &&  \sum_{n \in \mathcal{O}_L(k)} e^{j \left(\frac{N-1}{2}-n\right)  \frac{2\pi d}{\lambda}  \cos(\theta_\text{R}-\theta_\text{E}) } \bigg).
\end{eqnarray}
Since the entries of $\mathcal{I}_M(k)$ and $\mathcal{I}_L(k)$ are randomly selected for each data symbol, (\ref{anb}) can be simplified to
\begin{eqnarray}\label{bbeta}
  \beta &=& \sqrt{\frac{1}{{N_\text{T}}}} \sum_{n=0}^{N_\text{T}-1} W_n e^{j \left(\frac{N-1}{2}-n\right)  \frac{2\pi d}{\lambda} (\cos\theta_\text{R} - \cos\theta_\text{E})},
\end{eqnarray}
where $W_n$ is a Bernoulli random variable and $W_n = 1$ with probability $(M+(N_\text{T}-M)/2)/N_\text{T} = \frac{N_\text{T}+M}{2N_\text{T}}$, and $W_n = -1$ with probability $\frac{N_\text{T}-M}{2N_\text{T}}$.

\section{Performance Evaluation}\label{sec:ana}

In this section, we evaluate the performance of the proposed mmWave secure transmission technique. As shown in (\ref{arnl}), the received signal at an arbitrary eavesdropper is subject to both multiplication and additive noise. This makes the derivation of a simple closed form expression for the signal-to-noise ratio (SNR) and a  secrecy rate formula challenging, if not possible. To evaluate the performance of the proposed mmWave secure transmission technique, we define the the secrecy throughput $R$ (bits per channel use) as 
\begin{eqnarray}\label{Rs}
R = [\log_2(1+\gamma_{\text{R}}) -\log_2(1+\gamma_{\text{E}})]^+,
 \end{eqnarray}
where $\gamma_{\text{R}}$ is the SNR at the target receiver, $\gamma_{\text{E}}$  is the SNR at the eavesdropper, and $a^+$  denotes $\max\{ 0,a \}$.
From (\ref{c1}) and (\ref{y1b1s}), the SNR at the receiver can be expressed as
\begin{eqnarray}
\gamma_{\text{R}} &=& \frac{{P \alpha N_\text{R}}(\mathbb{E}[\mathbf{h}^*(\theta_\text{R})\mathbf{f}(k)])^2}{{P\alpha N_\text{R}}\text{var}[ \mathbf{h}^*(\theta_\text{R})\mathbf{f}(k)]+\sigma^2} \\ \label{snrg} &=& \frac{{P \alpha N_\text{R}} M^2 g_\text{R}^2 } {N_\text{T} \sigma^2} =\frac{{2P \alpha N_\text{R}} M^2  \sin^2 (\Phi_\text{R}/2)  } {N_\text{T} \sigma^2},
\end{eqnarray}  
where (\ref{snrg}) follows since the array gain is constant at the target receiver and $g^2_\text{R} = \frac{1}{2}(1-e^{-j\Phi_\text{R}})(1-e^{j\Phi_\text{R}})= (1-\cos(\Phi_\text{R})) =  2\sin^2 (\Phi_\text{R}/2) $, and $\Phi_\text{R}$ is the phase angle between the receiver's direct and the reflected path.

Similarly, the SNR at an eavesdropper can be expressed as
\begin{eqnarray}
\gamma_{\text{E}} &=& \frac{{P \alpha_\text{E} N_\text{E}}(\mathbb{E} [ \mathbf{h}^*(\theta_\text{E})\mathbf{f}(k)])^2}{{P\alpha_\text{E} N_\text{E}}\text{var}[ \mathbf{h}^*(\theta_{\text{E}})\mathbf{f}(k)]+\sigma_\text{E}^2} \\ \label{snrge} &=& \frac{{P \alpha_\text{E} N_\text{E}} g_\text{E}^2 (\mathbb{E}[\beta])^2}{{P\alpha_\text{E} N_\text{E} g_\text{E}^2 }\text{var}[ \beta]+\sigma_\text{E}^2},
\end{eqnarray}  
where $g^2_\text{E} =2\sin^2 (\Phi_\text{E}/2) $, and $\Phi_\text{E}$ is the phase angle between the eavesdropper's direct and the reflected path. To derive $\gamma_{\text{E}}$, we need to characterize the random variable $\beta$ in (\ref{bbeta}). Note that the random variable $\beta$ is a sum of complex random variable. Invoking the central limit theorem, the artificial noise term $\beta$ can be approximated by a complex Gaussian random variable for sufficiently large $N_{\text{T}}$, and, it can be completely characterized by its mean and variance. To complete the SNR derivation, we introduce the following lemma

\begin {lemma} \label{l1}
For large number of antennas $ N_{\text{T}}$ and $\theta_\text{R} \not = \theta_\text{E}$, 

\begin{eqnarray}\label{l1m}
\mathbb{E} [\beta] = {\frac{{M}}{N_{\text{T}}\sqrt{N_{\text{T}}}}}\frac{ \sin\left( {N_{\text{T}}} \frac{\pi d}{\lambda}(\cos \theta_{\text{R}} -\cos \theta_{\text{E}})  \right)} {\sin\left(  \frac{\pi d}{\lambda}(\cos \theta_{\text{R}} -\cos \theta_{\text{E}})  \right)},
\end{eqnarray}  
and 
 \begin{eqnarray}\label{l1v}
 \text{var} [\beta] \approx \frac{N_{\text{T}}^2-M^2}{N^2_{\text{T}}}.
\end{eqnarray}  
\end{lemma}
\begin{proof}
The proof can be found in \cite{me} and is omitted for space limitation.
\end {proof}

		\begin{figure}[t]
		\begin{center}
\includegraphics[width=3.5in]{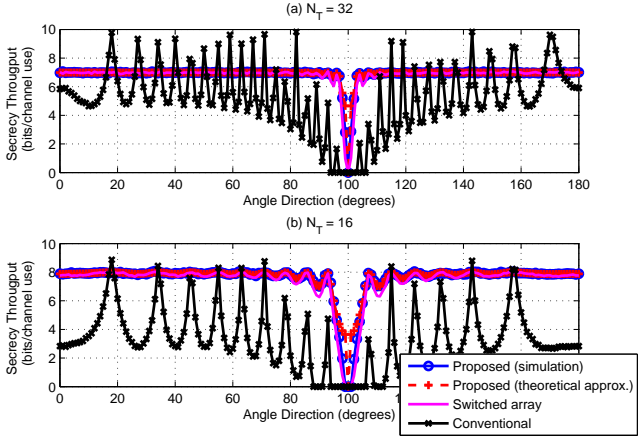}
				\caption{Numerical and theoretical (eq. (\ref{sar})) values of the secrecy throughput versus the eavesdropper's angular direction; (a) $N_\text{T}=32$ and $M=24$, (b) $N_\text{T}=16$ and $M=12$. }
			\label{fig1}
		\end{center}
	\end{figure}

Using  (\ref{l1m}) and (\ref{l1v}) $\gamma_\text{E}$ in (\ref{snrge}) becomes
 \begin{eqnarray}\label{snree2}
\nonumber \gamma_\text{E}  &\approx& \frac{1}{ \frac{N_{\text{T}}^2-M^2}{N^2_{\text{T}}}+\frac{\sigma^2_\text{E}}{2P\alpha_\text{E} N_\text{E} \sin^2(\Phi_\text{E}/2) }} \times \\ && \label{se} \bigg(  {\frac{{M^2}}{N^3_{\text{T}}}}\frac{ \sin^2\left( {N_{\text{T}}} \frac{\pi d}{\lambda}(\cos \theta_{\text{R}} -\cos \theta_{\text{E}})  \right)} {\sin^2\left(  \frac{\pi d}{\lambda}(\cos \theta_{\text{R}} -\cos \theta_{\text{E}})  \right)} \bigg),
\end{eqnarray}
where the approximation arises from $(\ref{l1v})$.
  
Substituting (\ref{snrg}) and (\ref{se}) in (\ref{Rs}), we obtain
\begin{eqnarray}\label{Rsb1}
\nonumber   R  &\approx&  \bigg[ \log_2\bigg(1+\frac{{2P \alpha N_\text{R}} M^2  \sin^2 (\Phi_\text{R}/2)  } {N_\text{T} \sigma^2} \bigg) -   \\ \nonumber &&  \log_2\bigg(1+\frac{1}{ \frac{N_{\text{T}}^2-M^2}{N^2_{\text{T}}}+\frac{\sigma^2_\text{E}}{2P\alpha_\text{E} N_\text{E} \sin^2(\Phi_\text{E}/2) }} \times \\ \label{sar} &&  \bigg(  {\frac{{M^2}}{N^3_{\text{T}}}}\frac{ \sin^2\left( {N_{\text{T}}} \frac{\pi d}{\lambda}(\cos \theta_{\text{R}} -\cos \theta_{\text{E}})  \right)} {\sin^2\left(  \frac{\pi d}{\lambda}(\cos \theta_{\text{R}} -\cos \theta_{\text{E}})  \right)} \bigg) \bigg)\bigg]^+.
 \end{eqnarray}

\section{Numerical Results and Discussions } \label{sec:PA}
In this section, we consider vehicle-to-vehicle transmission and study the performance of the proposed security technique in the presence of an eavesdropper. For comparison, we simulate and plot the secrecy throughput of conventional array (without PHY security) and switched array \cite{dm5} transmission techniques. Unless otherwise specified, the number of antennas at the transmitter is $N_\text{T} = 32$, $|g_\text{R}|^2 N_\text{R} = 8$, and $|g_\text{E}|^2N_\text{E} = 32$. The system operates at 60 GHz with a bandwidth of 50 MHz and an average transmit power of 37 dBm. The distance from the transmitter to the receiver is 30 meters, the distance from the transmitter to the eavesdropper is 10 meters, and the path loss exponent is fixed to 2.


In Fig. \ref{fig1}, we plot the secrecy throughput achieved by the proposed technique versus the eavesdropper's angular location when using 32 and 16 transmit antennas. We also plot the secrecy throughput achieved by conventional array and switched transmission methods. In Figs. \ref{fig1}(a) and (b), we  show that the secrecy throughput of the proposed technique is high at all angular locations except at the target receiver's angular location ($\theta_\text{R}=100$ degrees).  We also show that switched array techniques provide comparable secrecy throughput with the proposed technique, while conventional array techniques provide poor secrecy throughputs, especially for lower number of transmit antennas as shown in Fig. \ref{fig1}(b). The reason for this is that conventional array transmission techniques result in a constant radiation pattern at the eavesdropper while, the proposed and the switched array techniques randomize the radiation pattern at an eavesdropper, thereby creating artificial noise. For the proposed and switched array transmission techniques, no randomness is experienced at the target receiver, and the secrecy throughput is minimum when the eavesdropper is located in the same angle with the target receiver, which is generally not the case in practice. We also observe that the proposed technique performs better than switched array techniques especially when the eavesdropper is closely located to the legitimate receiver, which is typically the worst case scenario.

\begin{figure}[t]
		\begin{center}
\includegraphics[width=3.5in]{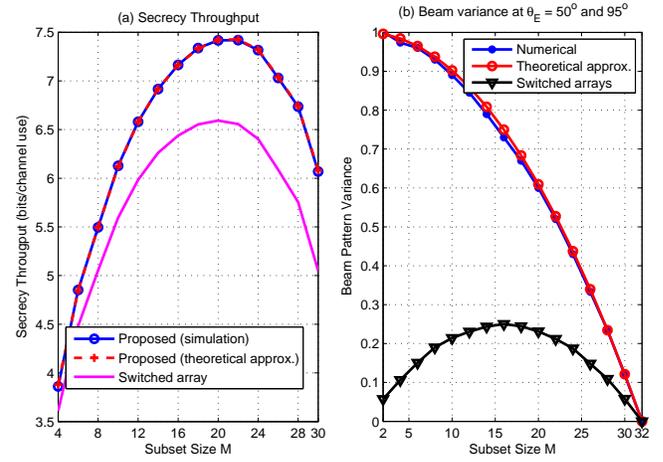}
				\caption{ (a) Numerical and theoretical (eq. (\ref{sar})) values of the secrecy throughput  versus the transmission subset size $M$; $N_\text{T}=32$, $\theta_\text{R} = 100^\circ$, and $\theta_\text{E} = 95^\circ$. (b) Beam variance at  $\theta_\text{E} = 50^\circ$ and $\theta_\text{E} = 95^\circ$ (superimposed).} 
			\label{fig2}
		\end{center}
	\end{figure}


In Figs. \ref{fig2}(a) and (b) we examine the impact of the transmission subset size $M$  on the secrecy throughput and the generated beam pattern variance (or artificial noise). In Fig. \ref{fig2}(a), we observe that as the subset size $M$ increases, the secrecy throughput of the proposed technique increases, plateaus, and then decreases. The reason for this is that as $M$ increases, more antennas are co-phased for data transmission. On one hand, this increases the beamforming gain at the target receiver. On the other hand, increasing $M$ decreases the variance of the artificial noise at the eavesdropper as shown in Fig. \ref{fig2}(b). Therefore, there is a trade-off between the beamforming gain at the receiver and the artificial noise at a potential eavesdropper, and there exists the optimum value of $M$ that maximizes the secrecy throughput. In Figs. \ref{fig2}(a) and (b), we also show that the proposed technique provides higher secrecy throughput and beam pattern variance (or artificial noise) when compared to the switched array technique without the need for additional antenna switches as required by the switched array technique. In the switched array technique, only $M$ antennas are used to generate the artificial noise, while the remaining $N_\text{T}-M$ antennas are idle. The proposed technique uses $N_\text{T}-M$ antennas to generate artificial noise at potential eavesdroppers, and hence, results in higher secrecy throughput and artificial noise. In addition to reducing the beam pattern variance at a potential eavesdropper, the idle antennas create a sparse array which could be exploited to launch attacks on switched array techniques. Since the proposed technique uses all antennas, it is difficult to breach. Finally, in Fig. \ref{fig2}(a), we show that the theoretical approximation of the secrecy throughput  is tight when compared to the secrecy throughput achieved by Monte Carlo simulation results. This enables us to gain further insights into the impact of key parameters such as, the subset size $M$ and number of transmit/receive antennas, on the system performance.

\section{Conclusions} \label{sec:con}

In this paper,  we proposed a PHY layer security technique for secure mmWave vehicular communication. The proposed technique takes advantage of the large antenna arrays at the mmWave frequencies to jam eavesdroppers with sensitive receivers. This enhances the security of the communication link between a  transmitter and a legitimate receiver.  The proposed technique is shown to provide high secrecy throughput when compared to conventional and switched array techniques  without the need for encryption/decryption keys and additional communication and processing overhead at the target receiver. This makes the proposed security technique  favorable for time-critical road safety applications and vehicular communications in general.

\section*{Acknowledgement} \label{sec:ackn}

This research was partially supported by the U.S. Department of Transportation through the Data-Supported Transportation Operations and Planning (D-STOP) Tier 1 University Transportation Center and by the Texas Department of Transportation under Project 0-6877 entitled “Communications and Radar-Supported Transportation Operations and Planning (CAR-STOP).

{}
\end{document}